\documentclass[a4paper,USenglish,numberwithinsect]{lipics-v2018}

\usepackage{latexsym,graphicx}
\usepackage{xspace}
\usepackage{ifpdf}
\usepackage{shadow,shadethm,color}
\usepackage{algorithm, algorithmic}

\definecolor{Darkblue}{rgb}{0,0,0.4}
\definecolor{Brown}{cmyk}{0,0.81,1.,0.60}
\definecolor{Purple}{cmyk}{0.45,0.86,0,0}

\ifpdf
 
\fi

\hideLIPIcs
\nolinenumbers

\newtheorem{observation}[theorem]{\bf Observation}
\newtheorem{claim}[theorem]{\bf Claim}
\newtheorem{appendix-claim}{\bf Claim A.\!\!}

\newcommand{\OPT}{\ensuremath{{\sf opt}}}
\renewcommand{\OPT}{\ensuremath{\textsf{\textsc{opt}}}}

\newcommand{\initOneLiners}{%
    \setlength{\itemsep}{0pt}
    \setlength{\parsep }{0pt}
    \setlength{\topsep }{0pt}
}

\newcommand{\fl}{{{\sc Facility Location}}\xspace}

\newcommand{\dist}{\ensuremath{\mathsf{dist}}}
\newcommand{\distlog}{\ensuremath{\mathsf{dist\_log}}}

\newcommand{\nqueries}{\ensuremath{n}} 
\renewcommand{\nqueries}{\ensuremath{\mathfrak{n}}}
\newcommand{\metric}{\ensuremath{\mathfrak{V}}} 
\renewcommand{\metric}{\ensuremath{\mathcal{V}}}
\newcommand{\nmetric}{N} 
\renewcommand{\nmetric}{\ensuremath{\mathfrak{m}}} 
\newcommand{\capac}{\ensuremath{\mathfrak{c}}} 
\newcommand{\nfin}{\ensuremath{\nqueries_{\mathfrak{act}}}} 
\newcommand{\hst}{\ensuremath{\mathfrak{T}}} 
\newcommand{\height}{\ensuremath{\mathfrak{h}}} 
\newcommand{\clients}{\ensuremath{\textsl{clients}}} 

\newcommand{\Artur}[1]{{{\footnote{\textcolor[rgb]{1.00,0.00,0.00}{{\sc\small \textbf{Artur:}} #1}}}}}

\renewcommand{\Artur}[1]{}

\EventEditors{Yossi Azar, Hannah Bast, and Grzegorz Herman}
\EventNoEds{3}
\EventLongTitle{26th Annual European Symposium on Algorithms (ESA 2018)}
\EventShortTitle{ESA 2018}
\EventAcronym{ESA}
\EventYear{2018}
\EventDate{August 20--22, 2018}
\EventLocation{Helsinki, Finland}
\EventLogo{}
\SeriesVolume{112}
\ArticleNo{21} 

\funding{Research partially supported
by the Royal Society International Exchanges Scheme 2013/R1, grant IE130346,
by the Centre for Discrete Mathematics and its Applications (DIMAP),
by EPSRC awards EP/D063191/1, EP/N011163/1, and
by the European Research Council research and innovation programme (ERC grant agreement No 677651).}
\title{Online Facility Location with Deletions
\Artur{``Thanks/support'' --- TO BE UPDATED}
}

\titlerunning{Online Facility Location with Deletions}

\author{Marek Cygan}{Institute of Informatics, University of Warsaw, Banacha 2, 02-097 Warsaw, Poland}{cygan@mimuw.edu.pl}{}{}
\author{Artur Czumaj}{Department of Computer Science and Centre for Discrete Mathematics and its Applications (DIMAP), University of Warwick, Coventry CV4 7AL, United Kingdom}{A.Czumaj@warwick.ac.uk}{}{}
\author{Marcin Mucha}{Institute of Informatics, University of Warsaw, Banacha 2, 02-097 Warsaw, Poland}{mucha@mimuw.edu.pl}{}{}
\author{Piotr Sankowski}{Institute of Informatics, University of Warsaw, Banacha 2, 02-097 Warsaw, Poland}{sank@mimuw.edu.pl}{}{}

\authorrunning{M.~Cygan, A.~Czumaj, M.~Mucha, P.~Sankowski}

\Copyright{Marek Cygan, Artur Czumaj, Marcin Mucha, Piotr Sankowski}


\begin{document}

\maketitle

\begin{abstract}
In this paper we study three previously unstudied variants of the online \fl problem, considering an intrinsic scenario when the clients and facilities are not only allowed to arrive to the system, but they can also depart at any moment.

We begin with the study of a natural \emph{fully-dynamic online uncapacitated} model where clients can be both added and removed. When a client arrives, then it has to be assigned either to an existing facility or to a new facility opened at the client's location. However, when a client who has been also one of the open facilities is to be removed, then our model has to allow to reconnect all clients that have been connected to that removed facility. In this model, we present an optimal $O(\log\nfin / \log\log\nfin)$-competitive algorithm, where $\nfin$ is the number of active clients at the end of the input sequence.

Next, we turn our attention to the \emph{capacitated} \fl problem. We first note that if no deletions are allowed, then one can achieve an optimal competitive ratio of $O(\log\nqueries / \log \log\nqueries)$, where $\nqueries$ is the length of the sequence.
However, when deletions are allowed, the capacitated version of the problem is significantly more challenging than the uncapacitated one. We show that still, using a more sophisticated algorithmic approach, one can obtain an online $O(\log \nmetric + \log\capac \log\nqueries)$-competitive algorithm for the capacitated \fl problem in the fully dynamic model, where $\nmetric$ is number of points in the input metric and $\capac$ is the capacity of any open facility.
\end{abstract}


\section{Introduction}

The \fl is one of the central combinatorial optimization problem, extensively studied in the literature for several decades, see, e.g., \cite{ASS17,byrka,GNW:1990,JV:2001,shi-li,STA:1997} and the references therein. The goal is to connect a given set of clients to a set of facilities such that the service cost is optimized. A very natural setting for the \fl problem is the \emph{online} scenario, where clients arrive incrementally over time and need to be connected to existing or newly opened facilities. Indeed, since it has been introduced over 15 years ago by Meyerson \cite{Meyerson:2001}, the online version of the \fl problem and its generalizations received considerable amount of attention \cite{Alon:2006,Anagnostopoulos2004175,AGLS16,dai2010,diveki2011,Fotakis2006,Fotakis2007,fotakis2008,Garg:2008,eyesclosed,kling2012,NW13,NET:NET20109,FWL16,umboh2015}. Typically in these models one assumes to be given a metric over a set of candidate points. When a client appears, we are allowed to open a new facility in her location, paying some cost for opening a new facility, and then we have to irrevocably connect the client to one of the open facilities, paying for the connection the cost equal to the distance from the client to the facility. In this paper, we study a more complex scenario and consider a natural extension of this model allowing clients to be \emph{removed} from the system, extending the standard dynamic model to the \emph{fully dynamic} setting. That is, in each time step either a new client arrives and needs to be connected to one of the open facilities, or one of the clients already existing in the system departs. Observe that if a client who has been also one of the open facilities is to be removed, then our model has to allow to reconnect all clients that have been connected to that removed facility.

\Artur{I think the model is good and natural, but it seems that we have been struggling with motivating it well. A good example is the following comment: \emph{``The topic itself is certainly interesting, although the two motivating examples given in the introduction (schools and P2P networks) do not seem to match the model very well.''} Any thoughts of how to sell it better? Maybe the example of the schools is not great? Any better suggestions?}The fully dynamic model is very natural in the context of the \fl problem, where in a number of scenarios it is desirable to dynamically process the arrival objects, and then to allow their departures. For example, if one wants to build and maintain schools in a newly developing city, one wants to allow a steady arrival of new pupils to the area, and also allow changes in the school demands when pupils population is declining. Similarly, if one wants to maintain a construction of a network, where all clients are to be connected to the servers (and pay connection costs) and each client is allowed to host a server (and pay an opening cost), one may also want to allow the removal of some clients and with that also a closure of some servers. Given that the relocation costs are often very expensive, the decisions should be regarded as irrevocable, unless, as in the case of deletions, are necessary. The framework considered in this paper is even more natural in complex distributed settings, for example, in some more modern scenarios that have been recently motivated by applications in peer-to-peer systems, when client and facility are coupled to the same entity. Consider the so called p2p networks with super-peers \cite{DBLP:conf/icde/YangG03} that were used for filesharing as Gnutella \cite{gnutella} or can be used for decentralized online social networks \cite{SD12}, distributed game systems \cite{x14}, grid management systems \cite{x5} and distributed storage \cite{x9}. In such case, some of the nodes decide to host the provided content. This decision incurs some cost to this node, but can reduce the cost of serving other nodes. The main question asked about such systems is about the needed ratio of superpeers to peers that guarantees that enough capacity of superpeers is available to serve all clients \cite{DBLP:conf/icde/YangG03}. This measure corresponds exactly to competitive ratio in our models.

Before discussing our results and techniques further, let us first formally define the models starting with the classical online variant.


\subsection{The model}

We study the performance of online algorithms for the \fl problem in the standard framework of \emph{competitive analysis} (cf. \cite{BE:1998,ST:1985}). A randomized algorithm for the \fl problem is \emph{$\alpha$-competitive} if for any input sequence, its expected cost is at most $\alpha$ times the optimal cost for the corresponding instance of the offline \fl. Note that the corresponding instance of the offline problem contains only the clients that are active at the end of the input sequence. 
We consider a standard special version of the \fl problem, where the set of clients and the set of possible facility locations are identical (see, e.g., \cite{BCIS05,Fot:2011,Meyerson:2001}).


\subsubsection{\emph{Online} facility location}
\label{subsec:online-fl}

We consider the \fl problem, where points (from some metric space)
arrive in online fashion. When a point $x$ arrives, we can first open a facility in $x$ and then we have to assign $x$ to some open facility. These choices are irrevocable. We consider only the \fl problem with \emph{uniform opening costs}, which are all, without loss of generality, equal to $1$. The restriction to the study of uniform costs makes the problem interesting, since without this assumption, no bounded competitive ratio can be obtained (cf. Claim A.\ref{claim:no-go-arbitrary-costs} in Appendix \ref{app:model}).

The cost of an obtained solution for a point set $\mathcal{X}$ is:
\begin{equation*}
    |\mathcal{F}| + \sum_{x \in \mathcal{X}} dist(x,\mathfrak{a}(x))
    \enspace,
\end{equation*}
where $\mathcal{F}$ is the set of open facilities and $\mathfrak{a}(x)$ is the open facility to which $x$ is assigned.
The cost of a point $x \in S$ is equal to its \emph{connecting cost} (distance to the open facility assigned to $x$) plus its \emph{opening cost} (if $x$ is an open facility then the opening cost is 1, and it is 0 otherwise).

The model defined here has been studied in the past, see e.g., Fotakis \cite{Fot:2011} for a survey or Meyerson \cite{Meyerson:2001}. (One frequently assumes (see e.g., \cite{Meyerson:2001}) that the opening cost in any facility is equal to some $f$, but simple scaling makes this problem equivalent to the one studied in our paper, that is, with $f=1$.)


\subsubsection{Online facility location with \emph{deletions}}

In this paper, we consider a \emph{fully-dynamic online} setting, in which in each time step either a new point (from some metric space) arrives in an online fashion, or a point already in the input is deleted. To cope with the case when an open facility is deleted, we have to allow the input points to be reassigned, and possibly to open other facilities. We consider the following, very natural model:
\begin{itemize}
\item In the online process, the requests are arriving online, and each request is either an arrival of a new demand point from a metric space, or a request to remove a previously inserted demand point.
\item When a new demand point $x$ arrives, the algorithm must at once and irrevocably decide if a new facility will open at $x$, and then must assign $x$ to some open facility (possibly to itself, if a facility was open at $x$).
\item When a point $x$ is requested to be removed, then $x$ is deleted from the system and if $x$ was an open facility, then all points assigned to $x$ will be immediately reassigned to other facilities and some of these points may become open facilities.
\end{itemize}


\subsubsection{Online \emph{capacitated} facility location}

We also consider a more general model of online \emph{capacitated \fl} with deletions, in which each open facility can handle at most $\capac$ clients, that is, at any moment at most $\capac$ clients can be connected to any single facility. Again, we study only a \emph{uniform} case (in which each facility has the same capacity $\capac$), since otherwise no bounded competitive ratio can be obtained (see Claim A.\ref{claim:no-go-arbitrary-capacities} in Appendix \ref{app:model}).


\bigskip

The goal of an algorithm for the \fl in any of the models defined above is to minimize the cost of the solution and to obtain an algorithm that is $\alpha$-competitive for $\alpha$ as small as possible. The performance of the algorithm may be a function of the \emph{length of the input sequence $\nqueries$}, the \emph{number of points $\nfin$ active} in a given moment, the \emph{size of the input metric space \nmetric}, and the \emph{capacity $\capac$}.

We note that in the model of uncapacitated \fl, the known results (as well as our new results) work for any metric space, even if it is \emph{unknown to the algorithm}. However, for the new algorithm for online capacitated \fl with deletions (cf. Theorem \ref{thm:capacitated-fl}), we will assume that \emph{the underlying metric is known to the algorithm in advance}.



\subsection{Related work}

As one of the fundamental problems in operations research and combinatorial optimization, the \fl problem has been studied extensively in the past, see, e.g., the standard exposition in \cite{GNW:1990} and more recent advances in \cite{ANS17,ASS17,JV:2001,STA:1997}\Artur{Have there been any new/recent papers we want to cite?}, and the references therein. In the online setting, an early research focused on the $k$-median problem (see, e.g., \cite{MP:2003}), which is a variant of the \fl problem where exactly $k$ facilities have to be opened.
%
Soon after, Meyerson \cite{Meyerson:2001} designed a simple randomized online algorithm in the uncapacitated model without deletion. His online algorithm is $O(\log \nqueries/\log\log \nqueries)$-competitive, where $\nqueries$ is the number of points in the input. (In fact, only a competitive ratio of $O(\log \nqueries)$ was proven in \cite{Meyerson:2001}, but Fotakis \cite{fotakis2008} extended the analysis from \cite{Meyerson:2001} to obtain a competitive ratio of $O(\log \nqueries/\log\log \nqueries)$.) Fotakis \cite{fotakis2008} later has shown that this bound is asymptotically tight and no online algorithm is $o(\log \nqueries/\log\log \nqueries)$-competitive; the lower bound holds for randomized algorithms against the oblivious adversary, for uniform facility costs, and for very simple metric spaces, such as the real line.
For more discussion about the history of the online version of the uncapacitated \fl problem (including deterministic online algorithms and incremental online algorithms), we refer to a survey by Fotakis~\cite{Fot:2011}.


The extension of the online model to deal with deletion of the facilities makes the \fl problem significantly more challenging. While this model is very natural, it requires a different approach that must permit to \emph{reverse some} of the decisions in the online algorithm,
and we are not aware of any study of algorithms for the online \fl problem in the fully dynamic setting, where deletions are allowed.


We note however, that similar fully dynamic models for optimization problems on graphs have been considered in the past, though only in limited settings.
For example, a number of graph optimization problems have been considered in fully dynamic models in the setting of \emph{data streams}, in the so-called \emph{turnstile model}. This model has been investigated in two scenarios: in the context of geometric graph optimization problems (see, e.g., \cite{czumaj2013,Indyk:2004,Lammersen2008}), and only very recently, in the context of standard graph optimization problems (see, e.g., a recent survey \cite{McG:2014} and the references therein). The main focus of these studies is to design algorithms that process a stream of data (in this case, edge or vertex insertions and deletions) and \emph{using very limited space}, to maintain some basic graph features.
For \emph{geometric graph optimization problems}, where the input is defined over a set of points in the discrete $d$-dimensional space $\{1, 2, \dots, \Delta\}^d$, it has been shown that many basic properties (e.g., $k$-median, minimum spanning tree, minimum weight matching, MaxCut) can be approximated very efficiently even with poly-logarithmic space (see, e.g., \cite{Indyk:2004,FS:2005}). The uncapacitated \fl problem has been studied in the context of data streaming, initiated with work of Indyk \cite{Indyk:2004}, who gave a $\text{poly}(\log \Delta)$-space streaming algorithm that approximates the optimal cost of the \fl problem within a factor of $O(\log^2 \Delta)$. The best currently known streaming algorithm using $\text{poly}(\log \Delta)$-space gives an $(1+\varepsilon)$-approximation for this problem \cite{czumaj2013}.
The research in data streaming for \emph{standard graph optimization problems} has been traditionally focusing on the insertion-only model, where one was aiming to design streaming algorithms with $O(n \, \text{poly} \log n)$ space (cf. \cite{McG:2014} and the references therein). Only very recently, Ahn et al.\ \cite{AGM:2012} initiated the study of algorithms that allow both insertions and deletions (see also \cite{CM:2005}). This line of research led to a number of efficient data streaming algorithms for fundamental graph problems, such as testing connectivity or bipartiteness, computing spanning trees and various graph sparsifiers, maximum matching, that can be (approximately) computed in small, $O(n \, \text{poly} \log n)$ space not only in the insertion only model, but also in the model with deletions \cite{AGM:2012,AGM:2013,KLMMS17,KW:2014,McG:2014}.

We also note that recently there has been some research on the standard online Steiner tree problem with deletions, see e.g., \cite{guptaK14,LOPSZ:2015}.\Artur{You (especially \textbf{\textcolor[rgb]{0.54,0.17,0.89}{Piotr}}) may prefer to elaborate more here.}


\subsection{New results}

The main contribution of this paper is the first thorough study of the online \fl problem with deletions and design of new algorithms for this model in several natural settings.


We begin with the study of the simplest, \emph{uncapacitated} model. We present in Theorem \ref{thm:unpacatitated-fl} an online $O(\log\nfin / \log\log\nfin)$-competitive algorithm for the uncapacitated \fl problem with deletions, where $\nfin$ is the number of active clients at the end of the input sequence; this bound gives an asymptotically optimal competitive ratio. Our algorithm is an extension of the classical insertion-only algorithm for the \fl problem due to Meyerson \cite{Meyerson:2001}, and we show that one can modify the analysis from \cite{Meyerson:2001} to allow deletions.

Next, we turn our attention to the \emph{capacitated} \fl problem. We first prove (Observation \ref{obs:Meyerson-capacitated}) that if \emph{no deletions} are permitted, then a simple modification of Meyerson's algorithm for online \fl can be applied to achieve an optimal competitive ratio of $O(\log\nqueries / \log \log\nqueries)$. However, when deletions are allowed, then the capacitated version of the problem is significantly more challenging than the uncapacitated one. We show that still, using more involved approach incorporating hierarchically well-separated trees, one can obtain an online $O(\log \nmetric + \log\capac \log\nqueries)$-competitive algorithm for the capacitated \fl problem with deletions, in the fully dynamic model, where $\nqueries$ is the number of queries, $\nmetric$ is the number of points in the input metric, and $\capac$ is the capacity of the facilities (Theorem \ref{thm:capacitated-fl}).

We notice that while the algorithms from Theorem \ref{thm:unpacatitated-fl} and Observation \ref{obs:Meyerson-capacitated} do not need to know the input metric, the result from Theorem \ref{thm:capacitated-fl} assumes that the input metric is known to the algorithm.

Our work demonstrates that despite the fact that the online \fl problem with deletion is clearly more complex than the classical online problem with insertions only, the most natural variants of these problems, for both the uncapacitated and the capacitated model for uniform facility costs, have very efficient online algorithms that achieve competitive ratios matching or almost matching those of the insertion only variants of the problem.


\section{Online uncapacitated facility location}

We begin with the study of the simplest, \emph{uncapacitated} model, and describe an insertion-only online algorithm for uncapacitated \fl due to Meyerson \cite{Meyerson:2001}.


\paragraph*{Algorithm M:}

When a new demand point $x$ arrives then find its nearest open facility $y$ and set $d_x = \min\{\dist(x,y), 1\}$.
Next, with probability $d_x$ open a new facility at point $x$ and assign $x$ to it; otherwise, assign $x$ to $y$.

\medskip


Meyerson \cite{Meyerson:2001} proved that Algorithm M is $O(\log \nqueries / \log\log\nqueries)$-competitive in the model with insertions only (see also \cite{fotakis2008,Fot:2011}).

The most natural approach to obtain an online algorithm for the \fl problem with deletions is to attempt to modify Algorithm M. Indeed, there is a simple modification of Algorithm M that addresses the deletions. The following Algorithm M* proceeds as in Algorithm M, except that when a facility is removed, it reprocesses all the points that are now not assigned to any facility using the original algorithm.


\paragraph*{Algorithm M*:}

When a new demand point $x$ arrives then find its nearest open facility $y$ and set $d_x = \min\{\dist(x,y), 1\}$.
Next, with probability $d_x$ open a new facility at point $x$ and assign $x$ to it; otherwise, assign $x$ to $y$.

When a point $x$ that is an open facility is to be deleted, then reassign all points assigned to $x$ using the algorithm for the insertions.

%

\medskip


While it is appealing to hope that Algorithm M* has performance similar to that of Algorithm M for insertions only,
but in fact asymptotically, it does no better than opening facilities at all points.

\begin{claim}
\label{example:1}
Algorithm M* has competitive ratio of $\Omega(\nfin)$.
\end{claim}


\begin{proof}
For any $k \in \mathbb{N}_+$, consider a star metric with center $\mathfrak{o}$ connected to points $\mathcal{X} = \{x_1, \dots, x_k\}$ at distance $\varepsilon=\frac{1}{k}$ from $\mathfrak{o}$. Consider the following input sequence:
\begin{enumerate}[1.]
\item Add $k^2$ clients $a_1, \dots, a_{k^2}$ located at $\mathfrak{o}$.
\item For each $i=1, \dots, k$ add a client $b_i$ located at $x_i$.
\item Remove the clients $a_1, \dots, a_{k^2-1}$ in that order.
\end{enumerate}

We will analyze the performance of algorithm $M^*$ in the above scenario. While doing that, we assume that whenever clients are reassigned, they are reassigned in the order in which they were originally assigned.

Consider the first two stages, when all the clients are added. In step 1, client $a_1$ opens a facility and then all other clients at $\mathfrak{o}$ connect to $a_1$. In step 2, each $b_i$ opens a facility with probability $\varepsilon$ and connects to $a_1$ otherwise. For the remainder of the analysis, let us assume that at least one of the clients $b_j$ opens a facility. This happens with probability $p_1 = 1 - (1 - \varepsilon)^k = \Omega(1)$.

Next, we remove $a_1$ and reassign all other clients assigned to $a_1$. One by one, each client at $\mathfrak{o}$ flips a coin and with probability $\varepsilon$ connects to one of the $b_i$'s; otherwise it opens a facility. As soon as some $a_i$ opens a facility, all other clients at $\mathfrak{o}$ connect to $a_i$, and each client at $\mathcal{X}$ that has not yet opened a facility (and hence was initially connected to $a_1$) does so with probability $\varepsilon$; otherwise it connects to $a_i$. If no $a_i$ opens a facility, then each client at $\mathcal{X}$ that has not yet opened a facility, opens one with probability $2 \varepsilon$, and otherwise it connects to one already open facility at some $b_j$.

We then remove all other clients $a_2, \dots, a_{k^2-1}$. Each time the client removed has a facility opened, the scenario described in the previous paragraph repeats. Removal of a client with no facility opened has, of course, no effect on the other clients.

Note that each time we remove a client at $\mathfrak{o}$ that has a facility open, each of the $b_j$'s opens a facility with probability at least $\varepsilon$, unless it has already opened a facility. How many times does that happen? Let us count the number of $a_i$'s with $2 \le i \le k^2-1$ that at some moment open a facility (and then get removed). Each of the $a_i$'s flips a coin exactly once (which corresponds to the situation that at some moment we have already removed $a_1, \dots, a_j$ for some $j < i$, and there is no facility open at $a_{j+1}, \dots, a_{i-1}$, with points $a_{j+1}, \dots, a_{i-1}$ connected to facilities from $\mathcal{X}$). Furthermore, the coin flip for $a_i$ is independent from coin flips made by $a_2, \dots, a_{i-1}, a_{i+1}, \dots, a_{k^2-1}$. Therefore the number of $a_i$'s that at some moment open a facility and then get removed has the binomial distribution with parameters $k^2-2$ and $ \varepsilon$. Hence, by Chernoff bound, with probability $p_2 \ge 1 - e^{-\varepsilon (k^2-2)/6} = \Omega(1)$, there are at least $\varepsilon (k^2-2)/2 \ge k/2-1$ of those clients. Therefore,
we can conclude that each $b_j$ opens a facility with probability at least $p_3 = 1-(1-\varepsilon)^{k/2-1} = \Omega(1)$.

The optimal offline solution opens a facility at $a_{k^2}$ and connects all points $b_1, \dots, b_k$ to $a_{k^2}$, and it has cost $1 + \varepsilon k = 2$. On the other hand, with probability $p_1 p_2 =\Omega(1)$ algorithm $M^*$ opens at least $kp_3 = \Omega(k)$ facilities from $b_1, \dots, b_k$ in expectation. This gives an $\Omega(k)=\Omega(\nfin)$ competitive ratio. So asymptotically $M^*$ does no better than just opening facilities at all points.
%
\end{proof}


\subsection{Asymptotically optimal competitive ratio for uncapacitated facility location with deletions}
\label{subsec:Algorithm1}

Claim \ref{example:1} shows one of the main challenges that online algorithms with deletions must cope with: we cannot let points attempt to open a facility too frequently, since then we would open too many facilities, and at the same time we have to be able to open some facilities in the neighborhood of the facilities that we are closing. To address this challenge we have to provide a delicate online strategy that will maintain a right balance between these two desirables.

\paragraph*{Algorithm 1:}
\textbf{Newly arriving points:} When a new demand point $x$ arrives then find its nearest open facility $y$ and set $d' = \min\{\dist(x,y), 1\}$.
With probability $d'$ open a new facility at point $x$ and assign $x$ to it; 
otherwise, assign $x$ to $y$ and set $p_x := d'$. (The algorithm will memorize $p_x$ for future use.)

\noindent\textbf{Deletion:} When a point that is \emph{not an open facility} is to be deleted, then just remove that point.

When a point that is an \emph{open facility} is to be deleted, then reassign all points assigned to it:
When a point $x$ is to be reassigned, then find its nearest open facility $y$ and set $d' = \min\{\dist(x,y), 1\}$. Let $p_x$ be the last value used in the processing of $x$. If $d' \le 2p_x$, then assign $x$ to $y$; otherwise with probability $d'$ open a new facility at $x$; else, with probability $1-d'$ assign $x$ to $y$ and set $p_x := d'$.

\medskip

The following theorem analyzes the performance of Algorithm 1.

\begin{theorem}
\label{thm:unpacatitated-fl}
Algorithm 1 is $O(\log\nfin / \log\log\nfin)$-competitive, where $\nfin$ is the number of active clients at the end of the input sequence (in particular, as $\nqueries$ is the input length, we have $\nfin \le \nqueries$).
\end{theorem}

\begin{proof}
We follow the approach of Meyerson \cite{Meyerson:2001}, with special care taken to deal with deletions.

Consider any optimal solution. For a fixed facility $v$, let $S$ be the set of clients connected to $v$ in the optimal solution. Let $d^*$ be the average distance from the points of $S$ to $v$. We split $S$ into $h+1$ subsets $S_0, S_1, \dots, S_{h}$, where $h = \lceil 2 \log |S| / \log \log |S| \rceil$. Points in $S_0$ are at distance at most $d^*$ from $v$. For $1 \le i \le h$, points in $S_i$ are at distance greater than $d^* (\log |S|)^{(i-1)/2}$, but at most $d^* (\log |S|)^{i/2}$ from $v$. Note that each point is contained in some $S_i$, as $d^* (\log |S|)^{h/2} \ge d^* |S|$, and a single client at distance more than $d^* |S|$ would contradict the average distance $d^*$.

For each set $S_i$ we split the time into two epochs: the second epoch starts when the first point in $S_i$ becomes an open facility in the solution of the algorithm --- note that the second epoch might never start.

Consider the first epoch of some $S_i$. The part of the cost of the final solution incurred by the points in $S_i$ during the first epoch
is the total connection cost of these points at the end of the first epoch plus possibly a single opening cost. The connection cost of a point $x \in S_i$ is upper bounded by twice the probability of the last coin flip of $x$, regardless of whether that connection was preceded by a coin flip or not. To bound the total connection cost of $S_i$ in the first epoch it is therefore enough to
bound the sum of the probabilities of all the coin flips related to $S_i$ made during that epoch.

\begin{lemma}
\label{lem:martingale}
The expected value of the sum of the probabilities of the coin flips related to points in $S_i$ and made before the first facility in $S_i$ opens is at most $1$.
\end{lemma}

\begin{proof}
Let $P_1,P_2,\ldots,P_M$ be the probabilities of all coin flips (with non-zero probabilities) made for points in $S_i$ and let $X_1,X_2,\ldots,X_M$ be the outcomes of these coin flips, so that $P[X_j=1] = P_j$ (note that each $P_j$ is a random variable, and not a constant since its value might possibly depend on earlier coin flips). Also, add to both sequences a virtual ,,sentinel'' coin flip with $P_{M+1}=1$, $X_{M+1}=1$. Define $Z_0 = 0$ and
\[ Z_{j+1} = Z_j + P_{j+1} - X_{j+1}\]
for $j=1,\dots,M+1$.
Then the sequence $\{Z_j\}$ forms a martingale. Let $T=\min \{ j>0 : X_j=1\}$ be the position of the first heads in $\{X_j\}$. Note, that $T$ is well defined, since $X_{M+1}=1$. Also note, that $T = \min \{ j>0 : Z_j \le Z_{j-1} \}$, i.e.,~$T$ is a stopping time for $\{Z_j\}$. Since $T$ is bounded, from the Doob's optional stopping time theorem
we get that $E[Z_T] = E[Z_0] = 0$. However, we also have
\[ Z_T = \sum_{j=1}^T P_j - 1,\]
and thus the claim follows. Note that the claim does not hold with equality, since the above expression might include the virtual coin flip added to make $T$ well defined.
\end{proof}

It follows that the total cost incurred by the points of $S_i$ during the first epoch is a constant, so the overall
cost of the first epoch is $O(h)$.

Consider now the cost incurred by any point $x \in S_i$ in the second epoch. If $x$ does not make any coin flips
in the second epoch, then any connection made by $x$ has cost upper bounded by $2d^*(\log |S|)^{i/2}$, since there
is an open facility in $S_i$ at this point. Consider now the case when $x$ does make a coin flip during the second epoch.
We then upper bound the cost of the last connection made by $x$ by twice the probability of the last coin flip of $x$.
We also need to bound the cost of a facility that $x$ might potentially open. The expected number of facilities opened
by $x$ is the sum of the probabilities of all its coin flips. Each term in this sum is at least twice the previous one,
and so the expected number of facilities opened by $x$ is at most twice the probability of the last coin flip of $x$.
This probability is at most $2d^*(\log |S|)^{i/2}$ since there is an open facility in $S_i$.

To sum up, the total cost incurred by $x \in S_i$ during the second epoch is at most $8d^*(\log |S|)^{i/2}$.
Consider first the case where $S_i$ is not the innermost layer, i.e., $0<i\le h$.
Any $x\in S_i$ is then connected to $v$ in the optimum solution,
and thus incurs a cost of at least $d^*(\log |S|)^{(i-1)/2}$.
Therefore, the cost incurred by such $x$ in the second epoch is at most $\sqrt{\log |S|}$ times the
corresponding cost in the optimal solution.
Consider now the innermost layer. For any client in $S_0$ Algorithm 1 pays at most $O(d^*)$,
leading to $O(d^* |S_0|)$ total cost of the second epoch, which by the definition of $d^*$
is at most a constant factor more than the connection cost payed by the optimal algorithm.

Note that in the analysis we completely ignore the cost generated by the clients,
that are removed during the course of the algorithm, as those clients in the end generate no cost to Algorithm 1.
\end{proof}

Since even in the incremental model (when points can only arrive, not vanish) there is an 
$\Omega(\log \nqueries/ \log \log \nqueries)$-lower bound (cf. \cite{fotakis2008}), Algorithm 1 is optimal up to a constant factor.


\section{Capacitated online facility location (with insertions only)}
\label{sec:FL-cap-with-ins}

Our result in Theorem \ref{thm:unpacatitated-fl} shows that for uncapacitated \fl with deletions, one can extend the approach from earlier works (see \cite{Meyerson:2001} and also \cite{fotakis2008}) to design online algorithms that achieve asymptotically optimal competitive ratio. However, the model of \emph{capacitated} \fl with deletions is significantly more complex. Still, in the most basic case, the model \emph{with insertions only}, we observe that Algorithm M can be extended to the model of capacitated \fl. We run Algorithm M with a single modification: the nearest facility $y$ now is the nearest facility that is not fully saturated, that is, that has still available capacity.

\begin{observation}
\label{obs:Meyerson-capacitated}
Algorithm M in the capacitated case is $O(\log \nqueries/\log\log \nqueries)$-competitive in the model with insertions only.
\end{observation}

\begin{proof}
We only sketch the proof, since it is a straightforward modification of the original proof of Meyerson \cite{Meyerson:2001} and Fotakis \cite{fotakis2008}. Similarly as described in the proof of Theorem~\ref{thm:unpacatitated-fl}, we partition the set $S$ of clients connected to some open facility $v$ into subsets $S_0, \ldots, S_h$, depending on their distance to~$v$. What is different from the analysis of Meyerson and Fotakis is the way in which we split time into epochs for each layer $S_i$, as here the first and second epoch can possibly interlace. Formally, the cost incurred by a client belongs to the first epoch, if at the moment when the client appears there is no unsaturated open facility in the layer $S_i$. If there is an unsaturated open facility in the layer $S_i$, then the cost incurred by the client is classified to the second epoch.

The total cost of clients of the first epoch is bounded by Lemma~\ref{lem:martingale}, and is at most $\ell + \nqueries/\capac$, where $\ell$ is the total number of sets $S_i$, as at most $\nqueries / \capac$ facilities may be saturated during the course of the algorithm. Note that $\ell$ is exactly $(h+1) = O(\log \nqueries/\log\log \nqueries)$ times greater than the number of open facilities in the optimal solution. Also, as each facility can serve only $\capac$ clients, the term $\nqueries / \capac$ is not greater than the cost of the optimal solution.

The cost of clients of the second epoch is bounded as in the proof of Theorem~\ref{thm:unpacatitated-fl}, i.e., Algorithm M pays at most $d^* (\log |S|)^{i/2}$ for each client from $S_i$, whereas optimal solution pays at least $d^* (\log |S|)^{(i-1)/2}$, assuming $i > 0$. The cost of clients from $S_0$ is bounded analogously as in the proof of Theorem~\ref{thm:unpacatitated-fl}.
\end{proof}


\section{Capacitated facility location with deletions}
\label{sec:FL-cap-with-del}

The result in Section \ref{sec:FL-cap-with-ins} may give a hope that also our result from Theorem \ref{thm:unpacatitated-fl} for uncapacitated \fl with deletions can be easily extended to the model of \emph{capacitated} \fl with deletions.
For example, let us think of a simple extension of Algorithm 1 to the capacitated case.
Perhaps the most natural idea is to let $d$ be the distance
to the closest open unsaturated facility, i.e., a facility
which still can serve additional clients.
Unfortunately this line of reasoning does not lead to a meaningful competitive ratio,
because in the case when all the clients arrive at the same location all
the distances are equal to zero and therefore such a modified version of Algorithm 1
would be deterministic.
Playing against a deterministic algorithm is very convenient for the adversary,
as the adversary might remove all the clients which were not turned
into open facilities by the algorithm, leading to $\Omega(\capac)$ competitive ratio.

One natural idea to introduce randomness to the algorithm is to increase the value of $d$, so that even if the distance to the closest facility is zero, the client might still decide to open a new facility. However, this idea alone does not seem to lead to any reasonable competitive factor. Below, we present a typical hard example for one possible implementation of this idea. (Note, that this competitive ratio is as bad as the one obtained by an algorithm that opens facilities in all input points.)

\begin{claim}
    Let $\mathcal{A}$ be Algorithm 1 modified, so that $d$ is the maximum of the distance to the closest unsaturated facility and $10/\capac$. Then, the competitive ratio of $\mathcal{A}$ is $\Omega(\capac)$.
\end{claim}

\begin{proof}
    Consider a star metric with the center $\mathfrak{o}$ and the remaining $10 \capac^2$ points at distance $\frac12$ from $\mathfrak{o}$. Let us analyze the performance of $\mathcal{A}$ on the following input sequence. First, we repeat $\capac$ times the following insertion: insert a point at $\mathfrak{o}$ and then $10 \capac$ points in different leaves of the metric, and then, at the end, remove all points located in the leaves.

    The optimal solution will have one open facility at one of the $\capac$ points located at $\mathfrak{o}$, and so the optimal cost is 1. We claim that the expected size of the solution found by $\mathcal{A}$ is $\Omega(\capac)$.

    To see this, consider a single round of insertions. If at the beginning of the round, there are no unsaturated facilities at $\mathfrak{o}$, then one is created with probability $\Theta(1)$. If that happens, then each of the clients inserted at the leaves opens a facility with probability $\frac12+10/\capac$ until the facility is saturated. Since there are $10\capac$ such clients, w.h.p.~the facility gets saturated. It also follows, that for any round, w.h.p.\ there are no unsaturated facilities at $\mathfrak{o}$ at the beginning of the round.

    The expected number of facilities opened at $\mathfrak{o}$ is the sum over all rounds of insertions of the probabilities that a facility is open at the beginning of the round. Based on the observations of the previous paragraph, this probability is $\Theta(1)$ for each round, and the claim follows.
\end{proof}


\subsection{Hierarchically well-separated trees and facility location}
\label{subsec:HSTs}

In the following Sections \ref{subsec:HSTs} -- \ref{subsec:Algorithm2-analysis}, we will present an online algorithm for the \emph{capacitated} \fl problem in the fully dynamic setting with low competitive ratio, with both insertions and deletions. We begin with a brief overview. We will assume that (unlike in the rest of the paper) \emph{the input metric is given in advance}, $\metric$ is the set of all points in the metric space, and $\nmetric$ is the number of points in the metric space, $\nmetric = |\metric|$. We use the embedding of the original metric space into hierarchically well-separated trees (cf. \cite{Bartal:1996,Bartal:1998,FRT:2004}), on which we run our \fl algorithm. Once we run the algorithm on a hierarchically well-separated tree $\hst$, for every open facility, we will evenly split its capacity into $\height = O(\log\capac)$ parts and then allocate each partial capacity solely to the points in one of $\height$ areas we will define later. Then we use a key property which ensures that every point $v$ that would use an open facility $u$ in the offline uncapacitated optimal solution, to find a replacement facility with still available capacity that is at the same distance from $v$ in $\hst$ as the distance from $u$ to $v$ in $\hst$.

Our approach relies on the concept of \emph{hierarchically well-separated trees (HSTs)}, which are metric spaces defined on the leaves of weighted rooted trees (cf. \cite{Bartal:1996,Bartal:1998,FRT:2004}). It is known (see \cite{FRT:2004}) that for every metric space, there is an HST with stretch $O(\log \nmetric)$. Let $\hst$ be such an HST for the metric given in our instance. Let the level of an internal node in the tree $\hst$ be the number of edges on the path to the root. Let $\Delta$ denote the diameter of the resulting metric space. In our paper, we will assume, without loss of generality, that the diameter of the original metric space is $\Delta = 1$. (Indeed, if a pair of points is at distance larger than 1, then in the \fl problem we will never allocate one of them to another, since it is always cheaper to pay 1 to open a new facility; therefore, we can treat any distance larger than 1 as equal to 1.)
We also modify short distances in the metric, that is,
for any pair of points in the metric we assume their distance is at least $1/\capac$.
Note that as $\nfin / \capac$ is a lower bound on the cost of the optimum solution
such modification of the metric increases the cost of the optimum solution at most
by a constant factor.
Then using the framework of HSTs, we will assume that the metric in the instance of capacitated online \fl we are solving is a shortest paths metric in a tree $\hst$ of depth $\height=\log \capac$, satisfying the following conditions:
\begin{itemize}
\item any edge connecting vertices of depth $i$ and $i+1$ is of length $2^{-i}$,
\item the set of potential clients are leaves of $\hst$,
\item all leaves of $\hst$ are at the same depth $\height$.
\end{itemize}


\begin{definition}
For two leaves $u$, $v$ of $\hst$, define $\distlog(u,v) = \lfloor - \log \dist_{\hst}(u,v) \rfloor$.
\end{definition}

\Artur{Reviewer wrote: \emph{``I think that $\distlog(u,v)$ is in fact the depth of the LCA minus 1. (Note that $\dist_{\hst}(u,v)$ is the length of the path in the tree from u to the LCA plus the length of the path in the tree from the LCA to $v$. In particular, this contains two edges of length $2^{-i}$, where $i$ is the depth of the LCA.)''}}Less formally, $\distlog(u,v)$ is the depth of the lowest common ancestor of $u$ and $v$ in $\hst$, which follows from the assumption that the weights of edges of $\hst$ are powers of two depending on their depth. From the triangle inequality we have the following property.

\begin{claim}
\label{claim:1}
For $u, v, w \in \metric$ we have $\distlog(u,w) \ge \min\{\distlog(u,v), \distlog(v,w)\}$.
\end{claim}


\subsection{Algorithm for fully dynamic capacitated facility location in HSTs}
\label{subsec:Algorithm2}

In this section, we present an algorithm for fully dynamic capacitated \fl in a hierarchically well-separated tree  $\hst$, Algorithm 2. We will analyze its performance in the next section.

We will assume the input metric space is the shortest path metric in a hierarchically well-separated tree $\hst$, with all input points coming from the leaves of $\hst$, as defined in the previous section. 
In the algorithm below, when a new facility is opened at a point $v$, then its capacity $\capac$ is evenly split into $\height$ parts, denoted as functions $cap_i(v)$ for $0 \le i < \height$. We will design the algorithm so that the capacity $cap_i(v)$ will be used solely by clients $u$ such that $\distlog(u,v) = i$, that is, by the clients such that the lowest common ancestor of $u$ and $v$ in $\hst$ is of depth $i$. Apart from that constraint, the algorithm mimics Algorithm~1 from Section \ref{subsec:Algorithm1}.

\paragraph*{Algorithm 2:}
%
\begin{itemize}
\item[] \textbf{insert} $v$:
        \begin{itemize}
        \item Call \textbf{connect}($v$).
        \end{itemize}
\item[] \textbf{delete} $v$:
        \begin{itemize}
        \item For all clients of $v$ (in arbitrary order) call \textbf{connect}($u$).
        \end{itemize}
\item[] \textbf{connect} $v$:
        \begin{itemize}
        \item If $\max_v$ is undefined, then set $\max_v = 0$.
        \item Let $u$ be the closest point to $v$ such that $cap_i(u) > 0$,
          where $i = \distlog(u,v)$.
        \item If such $u$ does not exist, then \textbf{open}($v$) and \textbf{exit}.
        \item $p = \min\{1, \dist_{\hst}(u,v) + \frac{12 \ \height \ln\nqueries }{\capac}\}$.
        \item If $p \le 2 \max_v$, then connect $v$ to $u$ and decrease $cap_i(u)$ by one.
        \item Otherwise: (i) set $\max_v = p$,
        (ii) with probability $p$ \textbf{open}($v$),
        and with probability $1-p$ connect $v$ to $u$ and decrease $cap_i(u)$ by one.
        \end{itemize}
\item[] \textbf{open} $v$:
        \begin{itemize}
        \item Open a facility at $v$.
        \item For each $0 \le i < \height$ set $cap_i(v) = \lfloor \capac / \height\rfloor$.
        \end{itemize}
\end{itemize}


\subsection{Key property of Algorithm 2}

Our analysis of the performance of Algorithm 2 begins with the following key lemma, which states that despite the capacity constraints, Algorithm 2 has the property that in every step of the algorithm, with high probability, every client will be able to connect to a nearest facility (in metric $\hst$).

\begin{lemma}
\label{lem:1}
Let $v \in \metric$ be a client that has never been deleted in the input sequence. If at some point $t_0$ of the course of Algorithm 2 a facility is opened at $v$, then with probability at least $1-\nqueries^{-3}$ from time point $t_0$ onwards, for any client $u$ there will always be an open facility $v' \in \metric$ (possibly $v'=v$), such that $\dist_{\hst}(u,v') \le \dist_{\hst}(u,v)$ and $cap_{\distlog(u,v')}(v') > 0$.
\end{lemma}

\begin{proof}
In order to prove the lemma, by the union bound, it is enough to show that for a fixed client $u$ and a fixed point of time $t_1 \ge t_0$ the claim holds with probability at least $1 - \nqueries^{-5}$.

Let us observe that if Algorithm 2 opens a facility at a client $x$, then the decrease of the value $cap_i(x)$ can only be caused by a client $y$ such that $\distlog(x,y) = i$.

For a fixed $u$ and $t_1$, let $S$ be the set of all clients $s$ that:
%
\begin{itemize}
\item were inserted before the time $t_1$,
\item have not been removed till time $t_1$, and
\item satisfy $\distlog(s,v) \ge \distlog(u,v)$ (i.e., $\dist_{\hst}(s,v) \le \dist_{\hst}(u,v)$).
\end{itemize}
%
Less formally, $S$ is the set of all \emph{active clients} at time $t_1$ located at the leaves of the subtree of $\hst$ rooted at the lowest common ancestor of $u$ and $v$.

Let us observe first that, by Claim \ref{claim:1}, for any $s \in S$ we have $\distlog(u,s) \ge \min\{\distlog(u,v),$ $\distlog(v,s)\} \ge \distlog(u,v)$,
and consequently,
\begin{align}
\label{eq:2}
    \dist_{\hst}(u,s) \le \dist_{\hst}(u,v)
    \enspace.
\end{align}
Therefore, any point from $s \in S$ that is an open facility with sufficient capacity at time $t_1$ is a suitable choice for $v'$ in the claim.

On the other hand, observe that no client $x$ outside of $S$ can be attached to a facility $s \in S$ and at the same time contribute to the decrease of $cap_k(s)$, for any $k \ge \distlog(u,v)$, as such a point would be a part of $S$, leading to a contradiction.

In view of these two observations, we only have to show that there is at least one open facility in $S$ with sufficient capacity at time $t_1$.

Let $\capac' = \lfloor \capac/h \rfloor$. Consider two cases. First, assume that $|S| < \capac'$. In that case there are not enough vertices to make $cap_{\distlog(u,v)}(v) = 0$, and hence the lemma follows for $v'=v$.

In the second case, assume that $|S| \ge \capac'$. Let us arbitrarily partition the set $S$ into groups $S_1, \dots, S_r$, such that each group is of size at least $\capac'/2$ and smaller than $\capac'$. Consider a single group $S_j$. Algorithm 2 ensures that any point of $S$ becomes an open facility with probability at least $\frac{12 \ln \nqueries}{\capac'}$. From the size of $S_j$ we infer that with probability at least
    $1 - \left(1 - \frac{12 \ln \nqueries}{\capac'}\right)^{|S_j|} \ge 1 - \left(1 - \frac{12 \ln \nqueries}{\capac'}\right)^{\capac'/2} \ge 1 - e^{-6 \ln \nqueries} = 1 - \nqueries^{-6}$
%
%
there is an open facility in $S_j$. Therefore, by the union bound, with probability at least $1 - \nqueries^{-5}$ there is at least one open facility in each of the groups $S_1, \dots, S_r$.
%
Hence, assuming that there is at least one open facility in each of the groups $S_1, \dots, S_r$, the total capacity provided by open facilities in $S$ is at least $r \cdot \capac' > |S|$, 
%
%
so there exists some open facility $v' \in S$, such that $cap_{\distlog(u,v')}(v') > 0$, where $\dist_{\hst}(u,v') \le \dist_{\hst}(u,v)$ by $(\ref{eq:2})$. This completes the proof of the lemma.
\end{proof}


\subsection{Fully dynamic capacitated facility location}
\label{subsec:Algorithm2-analysis}

With Lemma \ref{lem:1} at hand, we are now ready to conclude our analysis and present an $O(\log \nmetric + \log \capac \log \nqueries)$-competitive online algorithm for capacitated \fl with deletions. The algorithm assumes that the underlying metric space is known in advance.
Recall, that we have modified the initial metric so that its diameter is $\Delta = 1$ and
the minimum inter-point distance is at least $1/\capac$.

\begin{theorem}
\label{thm:capacitated-fl}
There is an $O(\log \nmetric + \log \capac \log \nqueries)$-competitive online algorithm for capacitated \fl with deletions.
\end{theorem}

\begin{proof}
We apply Algorithm 2 to an HST $\hst$ of depth $\height = \log \capac$ and
expected $O(\log \nmetric)$ stretch.

In what follows, we analyze the competitive ratio of this algorithm.
Let $C$ be the set of clients, which are still present at the end of the sequence of queries,
i.e., that were inserted by the adversary and have not been removed.
Let us consider an optimal offline solution $\OPT$ (for the original metric, not
with respect to the HST $\hst$).
Let $F_{\OPT}$ be the subset of facilities opened in $\OPT$ and for $f \in F_{\OPT}$, let
$\clients(f)$ be the set of clients served by $f$.
Note that the cost of $\OPT$ equals $|F_{\OPT}| + \sum_{f \in F_{\OPT}} \sum_{x \in \clients(f)} \dist(x,f)$.

Fix some $f \in F_{\OPT}$.
Similarly as in the earlier proofs, we are going to partition the set $\clients(f)$
into subsets, however this time the partition is not driven by distances of the metric,
but by the tree $\hst$.
We partition the set $\clients(f)$ into $\height$ subsets $S_1, \dots, S_{\height}$
such that for $x \in S_i$ we have $\distlog(x,f) = i$ (note that some $S_i$ might be empty).

For each $S_i$ we split the course of the algorithm into two epochs.
The second epoch starts as soon as the first point in $S_i$ becomes an open facility.
Note that, however, the second epoch might never start.

By the same analysis as in the uncapacitated case, the expected contribution of the
first epoch over all the facilities $f \in F_{\OPT}$ and all sets $S_i$
is upper bounded by $O(\height \cdot |F_{\OPT}|) = O(\log \capac \cdot \OPT)$.

Observe that by Lemma~\ref{lem:1},
with high probability, each client from $S_i$ (re)assigned in the second epoch has an available open unsaturated facility at distance at most $O(\dist_{\hst}(f, S_i))$. By the analysis from the proof of Theorem~\ref{thm:unpacatitated-fl}, expected cost incurred by a client in the second epoch for a fixed HST is upper bounded by $O(\dist_{\hst}(f, S_i)) + 12 \height \ln \nqueries / \capac$. In expectation (over HSTs) of the sum of the first terms over $\clients(f)$ is bounded by $O(\log \nmetric) \cdot \sum_{x \in \clients(f)} \dist(x,f)$. Overall the contribution of the first terms to the competitive ratio is $O(\log \nmetric)$.

As for the second terms of the form $12 \height \ln \nqueries / \capac$, their
sum over all the clients of $C$ is at most $\frac{12 \height \nfin  \ln \nqueries}{\capac}$,
but since $\OPT \ge \nfin / \capac$, the contribution of the second terms
to the competitive ratio is bounded by $O(\log \capac \cdot \log \nqueries)$.
\end{proof}


\section{Conclusions}

In this paper we present the first thorough study of natural variants of the online \fl problem, when the clients and facilities are not only allowed to arrive to the system, but they can also depart from the system at any moment. In this fully-dynamic online problem, we study two fundamental settings: uncapacitated and capacitated \fl for uniform facility costs.
For uncapacitated \fl, we design an extension of the classical insertion-only randomized algorithm for the \fl problem due to Meyerson \cite{Meyerson:2001}, and show that it achieves an asymptotically optimal competitive ratio of $O(\log\nfin / \log\log\nfin)$ (Theorem \ref{thm:unpacatitated-fl}).
The capacitated \fl is more complex, and here we first show (Observation \ref{obs:Meyerson-capacitated}) that if no deletions are allowed, then one can achieve an asymptotically optimal competitive ratio of $O(\log\nqueries / \log \log\nqueries)$, the same bound as it is known for the uncapacitated variant.
When deletions are allowed, the task is more challenging, but we still are able to incorporate the framework of hierarchically well-separated trees to obtain an online $O(\log \nmetric + \log\capac \log\nqueries)$-competitive algorithm for the capacitated \fl problem with deletions (Theorem \ref{thm:capacitated-fl}).

Our work demonstrates that despite the fact that the online \fl problem with deletion is clearly more complex than the classical online problem with insertions only, the most natural variants of these problems, for both the uncapacitated and the capacitated model for uniform facility costs, have very efficient online algorithms that achieve competitive ratios matching or almost matching those of the insertion only variants of the problem. It is an interesting open problem whether one can improve the competitive ratio for the capacitated case with deletions, in particular whether it is possible to remove the dependence on $\nmetric$.


\bibliographystyle{plainurl}
\bibliography{fl}


\appendix
\begin{center}\huge\bf Appendix \end{center}


\section{Requirements for the uniformity of the model}
\label{app:model}

As we have mentioned earlier in Section \ref{subsec:online-fl}, in this paper we consider the \fl problem with uniform opening costs and with uniform capacities of the facilities. We now argue, that with no assumptions on opening costs and capacities, no non-trivial competitive ratio is possible.

\paragraph*{Arbitrary costs.}
If the opening costs are arbitrary, then consider a simple two-client instance, where both clients are located at the same point, the first client has an opening cost of $M$, and the second one has a cost of $1$. The optimum solution has cost $1$, and an online algorithm has no way to avoid paying at least $M$, and $M$ can be arbitrarily large.

\begin{appendix-claim}
\label{claim:no-go-arbitrary-costs}
Any online algorithm for the \fl problem with arbitrary (e.g., non-uniform) opening costs has unbounded worst-case competitive ratio.
\end{appendix-claim}

\paragraph*{Arbitrary capacities.}
If the capacities are arbitrary, then consider an $M$-element sequence of clients, all located at the same point and with the same cost of $1$, where the first $M-1$ clients can only open a facility of capacity $1$, whereas the last client can open a facility of capacity $M$. The optimum offline solution only opens the facility for the last client and has a cost of $1$. An online algorithm is forced to open a facility for every client, for a total cost of $M$.

\begin{appendix-claim}
\label{claim:no-go-arbitrary-capacities}
Any online algorithm for the uncapacitated \fl problem with arbitrary (e.g., non-uniform) capacities has unbounded worst-case competitive ratio.
\end{appendix-claim}


\end{document}